\documentclass[longbibliography,aps,pra,superscriptaddress,twocolumn,nofootinbib]{revtex4-2}
\usepackage[a4paper,margin=23mm]{geometry}
\usepackage{graphicx}
\usepackage{amsmath,amssymb,amsthm,mathtools,bm}
\usepackage{braket}
\usepackage{booktabs}
\usepackage{microtype}
\usepackage{xcolor}
\usepackage{quantikz}
\usepackage[hidelinks]{hyperref}
\usepackage[nameinlink,noabbrev]{cleveref}
\usepackage{natbib}
\newtheorem{theorem}{Theorem}
\newtheorem{lemma}{Lemma}
\newtheorem{proposition}{Proposition}

\newcommand{\id}{\mathbb{I}}
\newcommand{\norm}[1]{\left\lVert#1\right\rVert}
\newcommand{\eps}{\varepsilon}
\newcommand{\GHZ}{\mathrm{GHZ}}
\newcommand{\sgnbar}{\overline{\operatorname{sgn}}}
\newcommand{\ketbra}[2]{\ket{#1}\!\bra{#2}}

\begin{document}

\title{Noise Robust Self Testing from Genuine Local Operation Shared Randomness multipartite nonlocality Tests}

\author{Som Kanjilal}
\email{somkanjilal@gmail.com}
\begin{abstract}
 Violation of the $N$-partite inequality introduced in Phys. Rev. Lett. 129, 150401 (2022) for genuine multipartite nonlocality under local operations and shared randomness (LOSR) rules out all causal-network models obtained by locally composing arbitrary resources involving at most $N-1$ parties, even when supplemented by shared randomness among all $N$ parties. Here, we demonstrate a device-independent self-test protocol for this. We further establish an analytic noise-robust self-test. For a Bell-score deficit $\epsilon$, we derive explicit $O(\sqrt{\epsilon})$ vector-norm bounds for the state and the corresponding measurement actions.
\end{abstract}

\maketitle

\section{Introduction}
Self-testing, introduced by Mayers and Yao, is a device-independent certification procedure in which the unknown state and measurements of arbitrary dimension implemented in an physical experiment are inferred solely from the observed input-output statistics \cite{MayersYao2004}. More precisely, a self-testing statement establishes that every quantum realization producing the specified correlations is equivalent to a reference realization up to local isometries. Unlike quantum tomography, this conclusion requires neither a bound on the local Hilbert-space dimensions nor a prior characterization of the measurement devices, although it assumes the validity of quantum theory and the spatial separation of the parties \cite{SupicBowlesReview}.

Exact self-testing statements are commonly associated with an ideal extremal correlation, often obtained from the maximal quantum violation of a Bell inequality. Experimental implementations, however, inevitably produce non-optimal correlations. Robust self-testing addresses this discrepancy by deriving quantitative bounds on the distance between the extracted realization and the reference realization as a function of the observed Bell-score deficit. Early works established robust self-testing statements for bipartite and multipartite states \cite{McKagueYangScarani2012,Mc2014,YangEtAl2014,BancalNavascues2015}. An important subsequent development was to provide noise robust self-testing methods which produced simple analytic and nearly optimal robustness bounds \cite{Bamps2015,Kaniewski2016}. Robust multipartite self-tests have since been derived from scalable Bell inequalities for graph states \cite{BaccariEtAl2020}, while GHZ states and local incompatible measurements have been certified using, among others, chained Bell inequalities \cite{SupicEtAl2016,PaulEtAl2025}, graph-state inequalities \cite{BaccariEtAl2020}, and Svetlichny-type inequalities \cite{SinghSasmalPan2025,PaulAdhikaryPan2026}.

Most of these constructions are formulated relative to either standard Bell locality or a Svetlichny-type notion of genuine multipartite nonlocality. The latter excludes hybrid models in which, in each run, the parties are separated into groups that may share arbitrary correlations internally. A different, theory-agnostic notion of genuine multipartite nonlocality was introduced in \cite{CoiteuxRoyWolfeRenou2021PRL,CoiteuxRoyWolfeRenou2021PRA,Cao2022}. An $N$-partite correlation is called genuinely Local Operation Shared Randomness (LOSR) multipartite nonlocal if it cannot be generated by local operations acting on an arbitrary collection of causal generalized-probabilistic-theory resources, each shared by at most $N-1$ parties, even when all $N$ parties have access to unrestricted global shared randomness\cite{Tavakoli_2022}.

Mao \emph{et al.}\cite{MaoEtAl2022} constructed a family of Bell-type inequalities capable of witnessing this form of genuine multipartite nonlocality . The $N$-party inequality involves two binary measurements per party and only $N+2$ correlators, and its quantum maximum is attained by the $N$-qubit GHZ state with the corresponding stabilizer measurements. Violation of its LOSR bound proves that the observed correlations cannot be simulated by causal resources involving at most $N-1$ parties\cite{Tavakoli_2022}. Such a separation, however, does not by itself establish self-testing. In particular, it does not imply that the quantum state and measurements attaining the maximal value are essentially unique, nor does it quantify what can be certified when the observed value is below the quantum maximum.

The exact self-testing property of this functional is closely connected to the stabilizer construction of Zhao and Zhou
\cite{Zhao2022}. A suitable choice of GHZ stabilizers in their construction yields the same correlator structure, with different positive weights on the unpaired stabilizer terms. Since the constituent quantum bounds can be saturated simultaneously, this reweighting preserves the exact maximizing conditions. Our exact analysis therefore specializes their self-testing construction to the genuine LOSR witness of Ref.~\cite{MaoEtAl2022}. We give an explicit sum-of-squares decomposition and local SWAP isometry that form the basis of the robustness analysis.

Zhao and Zhou also investigated robustness numerically for selected graph states of three to six qubits \cite{Zhao2022}. Here, we derive an explicit analytic robustness bound for the Mao functional valid for every $N\geq3$.
Let $\eps=\beta_N^{\rm Q}-\langle\mathcal S_N\rangle$ denote the absolute Bell-score deficit, where
$\beta_N^{\rm Q}=2\sqrt{2}+2(N-2)$, and let $F_N=\langle\mathrm{GHZ}_N|\rho_{\rm sw}|
\mathrm{GHZ}_N\rangle$ be the fidelity of the extracted state.
We prove
\[
F_N\geq\max\{0,1-\kappa\eps\},
\qquad
\kappa=\frac{3+\sqrt{5}}{\sqrt{2}},
\]
with a coefficient independent of $N$ for this normalization of the Bell functional. We further derive explicit
$O(\sqrt{\eps})$ vector-norm bounds for the state and for individual and joint measurement actions. These results turn
a witness of genuine LOSR multipartite nonlocality into a quantitative, noise-robust device-independent certification
protocol.

\section{Preliminaries of Self-Testing}

In self-testing we distinguish between a \emph{physical experiment} and a \emph{reference experiment}. The physical experiment is the actual black-box experiment: it consists of an unknown state
\begin{equation}
\rho\in\mathcal{D}\left(\bigotimes_{i=0}^{N-1}\mathcal{K}_i\right),
\end{equation}
of arbitrary local dimension from the set of density operators of Hilbert space $\bigotimes_{i=0}^{N-1}\mathcal{K}_i$, $\mathcal{D}\left(\bigotimes_{i=0}^{N-1}\mathcal{K}_i\right)$.

We consider a purification of $\rho$ given by
\begin{equation}
\ket{\psi} \in \mathcal{D}\left(\left(\bigotimes_{i=0}^{N-1}\mathcal{K}_i\right)\otimes\mathcal{K}_{E}\right), \quad \operatorname{Tr}_{E}[\ketbra{\psi}{\psi}]=\rho.
\end{equation}
 The reference experiment consists of a specified target state
\begin{equation}
\ket{\psi^{\mathrm{ref}}}
\in\bigotimes_{i=0}^{N-1}\mathcal{H}_i^{\mathrm{ref}}
\end{equation}
and specified reference measurements $M_{i,x}^{\mathrm{ref}}$ acting on $\mathcal{H}_i^{\mathrm{ref}}$. In the present work,
$\mathcal{H}_i^{\mathrm{ref}}\cong\mathbb C^2$ and the target state will be the $N$-qubit GHZ state.

In the physical experiment, party $i$ receives a classical input $x_i\in\{0,1\}$ and performs a binary measurement with outcome $a_i\in\{0,1\}$. By applying local Naimark dilations to the measurements, we may, without loss of generality,  represent every binary measurement by a Hermitian reflection
\begin{equation}
O_{i,x_i}=O_{i,x_i}^{\dagger},
\qquad
O_{i,x_i}^{2}=\id_{\mathcal K_i}.
\label{eq:physical-reflections}
\end{equation}
Measurements belonging to distinct parties act on distinct tensor factors and therefore commute:
\begin{equation}
[O_{i,x_i},O_{k,x_k}]=0, \qquad i\neq k.
\label{eq:cross-party-commutation}
\end{equation}
The projectors associated with the two outcomes are
\begin{equation}
\Pi_{a_i|x_i}^{(i)}:=\frac{\id_{\mathcal{K}_i}+(-1)^{a_{i}}O_{i,x_i}}{2},
\qquad a_i\in\{0,1\}.
\label{eq:physical-projectors}
\end{equation}
For input and output strings
\begin{equation}
\mathbf{x}=(x_0,\ldots,x_{N-1}),\qquad\mathbf{a}=(a_0,\ldots,a_{N-1}),
\end{equation}
the observed behavior is
\begin{equation}
p(\mathbf{a}|\mathbf{x})=\bra{\psi}\bigotimes_{i=0}^{N-1}\Pi_{a_i|x_i}^{(i)}\ket{\psi}.
\label{eq:observed-behavior}
\end{equation}
In particular, the full correlator associated with
$\mathbf{x}\in\{0,1\}^{N}$ is
\begin{equation}
    E(\mathbf{x}):=\bra{\psi}
    \bigotimes_{i=0}^{N-1}O_{i,x_i}
    \ket{\psi}.
    \label{eq:full-correlator}
\end{equation}
More generally, the behavior also determines correlators associated with every subset of the parties.

A self-testing statement asserts that every physical realization producing the specified correlations is equivalent to the reference realization up to local changes of basis and unused local degrees of freedom. More precisely, there exist local isometries
\begin{equation}
    V_i:
    \mathcal{K}_i
    \longrightarrow
    \mathcal{K}_i^{\xi}
    \otimes\mathcal{H}_i^{\mathrm{ref}},
    \quad
    V:=\bigotimes_{i=0}^{N-1}V_i,
    \label{eq:local-isometries}
\end{equation}
and a normalized, possibly entangled, auxillary state
\begin{equation}
\ket{\xi}\in\mathcal{D}\left(\left(\bigotimes_{i=0}^{N-1}\mathcal K_i^{\xi}\right)\otimes\mathcal{K}_{E}\right)
\end{equation}
such that
\begin{equation}
(V\otimes\id_{E})\ket{\psi}=\ket{\xi}\otimes\ket{\psi^{\mathrm{ref}}}.
 \label{eq:selftest-state-definition}
\end{equation}

To self-test the measurement actions as well as the state, we require that, for every subset $J\subseteq\{0,\ldots,N-1\}$ and every choice of input $x_i$ for $i\in J$,
\begin{equation}
\begin{split}
(V\otimes\id_E)\left[\left(\bigotimes_{i\in J}O_{i,x_i}\right)\ket{\psi}\right]
    =
    &\ket{\xi}\otimes\\
    &\left(
       \bigotimes_{i\in J}O_{i,x_i}^{\mathrm{ref}}
    \right)
    \ket{\psi^{\mathrm{ref}}},  
\end{split}
\label{eq:selftest-measurement-definition}
\end{equation}
where identity operators on the parties outside $J$ are understood. The case $J=\varnothing$ reduces to Eq.~\eqref{eq:selftest-state-definition}.

The normalized auxiliary state $|\xi\rangle$ belong to the Hilbert space $\left(\bigotimes_{i=0}^{N-1}\mathcal K_i^{\xi}\right)\otimes\mathcal{K}_{E}$. We identify output spaces up to the fixed permutation placing all auxiliary factors before the reference qubits. Whenever $V$ acts on a purified vector, $V\otimes\id_E$ is understood. The notation $\id_\xi$ and $\operatorname{Tr}_\xi$ refers to the full auxiliary space, including $E$.

From now on, we will suppress the identity operator $\id_E$ for the sake of brevity. Every local isometry $V_i$ admits a unitary dilation $\Phi_i$ after appending a sufficiently large local ancilla. In particular,
\begin{equation}
    V_i\ket{\phi}
    :=\Phi_i\left(\ket{\phi}\otimes\ket{0}_{a_i}\right).
    \label{eq:isometry-unitary-extension}
\end{equation}
If necessary, the local ancillary space may be enlarged so that $\Phi_i$ admits a unitary extension. Defining
\begin{equation}
    \Phi:=\bigotimes_{i=0}^{N-1}\Phi_i,
    \quad
    \ket{0^N}_{a}:=\bigotimes_{i=0}^{N-1}\ket{0}_{a_i},
\end{equation}
the state self-testing relation of \cref{eq:selftest-state-definition} becomes
\begin{equation}
    \Phi\left(\ket{\psi}\otimes\ket{0^N}_{a}\right)
    =
    \ket{\xi}
    \otimes\ket{\psi^{\mathrm{ref}}}.
    \label{eq:unitary-selftest-state}
\end{equation}
Similarly, for every subset
$J\subseteq\{0,\ldots,N-1\}$, the measurement self-test of \cref{eq:selftest-measurement-definition} becomes
\begin{equation}
\begin{split}
    &\Phi\left[
       \left(
         \bigotimes_{i\in J}O_{i,x_i}
       \right)\ket{\psi}
       \otimes\ket{0^N}_{a}
    \right] \\                                                 \\
    &\qquad=
    \ket{\xi}\otimes
    \left(
       \bigotimes_{i\in J}O_{i,x_i}^{\mathrm{ref}}
    \right)
    \ket{\psi^{\mathrm{ref}}}.
\end{split}
\label{eq:unitary-selftest-measurements}
\end{equation}

Note that, the self-testing theorems use the projective representation on enlarged local spaces obtained by the Naimark dilation; transferring them to undilated POVM actions would require an additional argument \cite{SupicBowlesReview}.

\section{Genuine LOSR multipartite nonlocality Test}  

Consider a scenario involving $N\geq 3$ space-like separated parties. For each party, there are two inputs and for each input there are two possible outputs. The $y$th measurement of the first party is $A_{0,y}$ where $y\in\{0,1\}$. The $j$th measurement of the $i$th party where $i\geq 1$ and $j\in\{0,1\}$ is denoted as $M_{i,j}$. The genuine (LOSR) nonlocality functional of \cite{MaoEtAl2022} is then given by
\begin{equation}
 \label{eq:SN-alt}
\begin{split}
\mathcal S_N&=
 (A_{0,0}+A_{0,1})M_{1,0}+ (A_{0,0}-A_{0,1})\prod_{k=1}^{N-1}M_{k,1}\\
 &+2\sum_{k=1}^{N-2}M_{k,0}M_{k+1,0},
\end{split}
\end{equation}
where we suppressed the tensor product.
The LOSR-local bound is $\langle\mathcal S_N\rangle\le2(N-1)$, while the
quantum maximum is
\begin{equation}
 \beta_N^{\rm Q}=2\sqrt{2}+2(N-2).
 \label{eq:qmax}
\end{equation}
One realization of the maximum is
\begin{align}
 \ket{\GHZ_N}&=\frac{\ket{0^N}+\ket{1^N}}{\sqrt{2}},\\
 M_{k,0}^{\mathrm{ref}}&=\sigma_z, \quad M_{k,1}^{\mathrm{ref}}=\sigma_x,\quad (k\geq 1)\nonumber\\
 A_{0,0}^{\mathrm{ref}}&=\frac{\sigma_z+\sigma_x}{\sqrt{2}},
 \quad A_{0,1}^{\mathrm{ref}}=\frac{\sigma_z-\sigma_x}{\sqrt{2}}.
 \label{eq:reference-realization}
\end{align}
Equation~\eqref{eq:reference-realization} is a reference realization, not an assumption about the black boxes. For the first party (Bob) let us define
\begin{equation}
 A_+:=\frac{A_{0,0}+A_{0,1}}{\sqrt{2}},\qquad
 A_-:=\frac{A_{0,0}-A_{0,1}}{\sqrt{2}}.
 \label{eq:Hpm}
\end{equation}
They satisfy the operator identities
\begin{align}
 A_+^2+A_-^2&=2\id,\label{eq:Hsquare-sum}\\
 \{A_+,A_-\}&=\frac{1}{2}\{A_{0,0}+A_{0,1},A_{0,0}-A_{0,1}\}\nonumber\\
 &=A_{0,0}^2-A_{0,1}^2=0.
 \label{eq:H-anti}
\end{align}
Note that, $A_{\pm}$ need not be reflections. By defining
\begin{equation}
 P:=\prod_{k=1}^{N-1}M_{k,1},
 \label{eq:P-def}
\end{equation}
the nonlocality operator becomes
\begin{equation}
\begin{split}
\mathcal S_N & =
\sqrt{2}A_+M_{1,0} +\sqrt{2}A_-P\\
&+2\sum_{k=1}^{N-2}M_{k,0}M_{k+1,0}.
\end{split}
 \label{eq:SN-compact}
\end{equation}
The self testing proof is constructed by showing that every physical realization 
$$\{\ket{\psi},M_{k,i},A_{0,j}:i,j\in\{0,1\},k\in\{1,2,\ldots,(N-1)\}\}$$ 
attaining the quantum maximum satisfies \cref{eq:unitary-selftest-state} and \cref{eq:unitary-selftest-measurements} for the reference realization given by \cref{eq:reference-realization}.

\section{Self-testing with the genuine-network inequality}

The exact proof follows the standard methodology of other self-testing proofs. The sum-of-squares (SOS) decomposition first gives operator-on-state identities. Using those operator identities we prove the SWAP isometry. 

\subsection{The Sum-of-Squares (SOS) decomposition}

Let us define the two Bob-coupled SOS operators as
\begin{equation}
 \Delta_+:=M_{1,0}-A_+,\qquad \Delta_-:=P-A_-,
 \label{eq:delta-def}
\end{equation}

\begin{proposition}[Fixed SOS identity]
For genuine $N$ party nonlocality functional given by \cref{eq:SN-compact}, the following is true
\begin{equation}
\begin{split}
\beta_N^{\rm Q}\id-\mathcal S_N
 &=\frac{1}{\sqrt{2}}(\Delta_+^2+\Delta_-^2)\\
 &+\sum_{k=1}^{N-2}(M_{k,0}-M_{k+1,0})^2,
\end{split}
 \label{eq:SOS}
\end{equation}
where $\Delta_{\pm}$ is given by \cref{eq:delta-def} and $\beta_N^{\rm Q}$ is given by \cref{eq:qmax}.
\end{proposition}
\begin{proof}
Using \cref{eq:Hsquare-sum}, $M_{1,0}^2=\id$, and
$P^2=\id$, we get
\begin{align}
 \frac{1}{\sqrt{2}}(\Delta_+^2+\Delta_-^2)
 &=2\sqrt{2}\id-\sqrt{2}A_+M_{1,0}\nonumber\\
 & -\sqrt{2}A_-P.
 \label{eq:SOS-first}
\end{align}
The $N-2$ remaining operators on the right hand side of \cref{eq:SOS} gives
\begin{align}
 \sum_{k=1}^{N-2}(M_{k,0}-M_{k+1,0})^2&=2(N-2)\id\nonumber\\
 & -2\sum_{k=1}^{N-2}M_{k,0}M_{k+1,0}.
 \label{eq:SOS-chain}
\end{align}
Adding \cref{eq:SOS-first,eq:SOS-chain} yields \cref{eq:SOS}.
\end{proof}
Considering
\begin{equation}
 \Lambda_k:=M_{k,0}-M_{k+1,0}, \quad (1\leq k\leq N-2),
 \label{eq:lambda-def}
\end{equation}
 for a quantum state $\ket{\psi}$ \cref{eq:SOS} can be written as 
\begin{equation}
\label{eq:SOSexpect}
\beta_N^{\rm Q}-\langle\mathcal S_N\rangle
 =\frac{1}{\sqrt{2}}(\norm{\Delta_+}_{\psi}^2+\norm{\Delta_-}_{\psi}^2)
 +\sum_{k=1}^{N-2}\norm{\Lambda_k}_{\psi}^2
\end{equation}
where $\norm{X}_\psi:=\norm{X\ket{\psi}}=\sqrt{\langle \psi|X^{\dagger}X|\psi\rangle}$. Since every term on the right hand side of \cref{eq:SOSexpect} is a nonnegative number, we have the following result 
\begin{lemma}
\label{eq:SOS-operator-identity-for-max}
at maximal violation of the nonlocality functional of \cref{eq:SN-compact} we have 
\begin{equation}
 \Delta_+\ket{\psi}=\Delta_-\ket{\psi}
 =\Lambda_k\ket{\psi}=0,
 \label{eq:residual-annihilate}
\end{equation}
for all $1\leq k\leq N-2$.
\end{lemma}
Now, we will find the local operator identities for a quantum strategy $\{\ket{\psi},A_{\pm}, (M_{i,0},M_{i,1}):i\geq 1\}$ which gives the exact maximum quantum violation of \cref{eq:SN-compact} by satisfying \cref{eq:residual-annihilate} and yielding the value $\beta_N^{Q}$.

\subsection{Operator identities for the optimal quantum strategy}

Let us rewrite \cref{eq:residual-annihilate} as
\begin{align}
 M_{1,0}\ket{\psi}&=A_+\ket{\psi},\quad
 P\ket{\psi}=A_-\ket{\psi},
 \label{eq:raw-exact}\\
 M_{1,0}\ket{\psi}&=M_{2,0}\ket{\psi}=\cdots=M_{N-1,0}\ket{\psi}.
 \label{eq:S-chain-exact}
\end{align}
Because $A_{\pm}$ commutes with $M_{1,0}$ and $P$ we can rewrite \cref{eq:raw-exact} as
\begin{align}
 A_+^2\ket{\psi}&=A_+M_{1,0}\ket{\psi}=M_{1,0}A_+\ket{\psi}=\ket{\psi},\label{eq:Hplus-onstate}\\
A_-^2\ket{\psi}&=A_-P\ket{\psi}=PA_-\ket{\psi}=\ket{\psi}.
\label{eq:Hminus-onstate}
\end{align}
 Inserting $A_{\pm}^2=\id\pm\{A_{0,0},A_{0,1}\}/2$ into
\cref{eq:Hplus-onstate,eq:Hminus-onstate} we get the local complementary relation for the first party
\begin{equation}
 \{A_{0,0},A_{0,1}\}\ket{\psi}=0,
 \label{eq:Bob-raw-anti}
\end{equation}
We now derive local complementarity relations for the rest.  Using
\cref{eq:raw-exact,eq:H-anti} and only cross-party commutation,
\begin{align}
 M_{1,0}P\ket{\psi}
 &=M_{1,0}M_{1,1}\left(\prod_{k\geq 2}M_{k,1}\right)\ket{\psi}\nonumber\\
 &\overset{\eqref{eq:raw-exact}}{=}M_{1,0}A_-\ket{\psi}\nonumber\\
 &=A_-M_{1,0}\ket{\psi}\overset{\eqref{eq:raw-exact}}{=}A_-A_+\ket{\psi}\nonumber\\
 &\overset{\eqref{eq:H-anti}}{=}-A_+A_-\ket{\psi}\overset{\eqref{eq:raw-exact}}{=}-A_+P\ket{\psi}\nonumber\\
 &=-PA_+\ket{\psi}\overset{\eqref{eq:raw-exact}}{=}-PM_{1,0}\ket{\psi}\nonumber\\
 & = -\left(\prod_{k\geq 2}M_{k,1}\right)M_{1,1}M_{1,0}\ket{\psi},\nonumber\\
\label{eq:Alice-anti-with-P}
\{M_{1,0},M_{1,1}\}&\left(\prod_{k\geq 2}M_{k,1}\right)\ket{\psi}=0.
\end{align}
Since the product term in \cref{eq:Alice-anti-with-P} commutes with $M_{1,0},M_{1,1}$ and $M_{k,1}^2=\id$, multiplying both sides of \cref{eq:Alice-anti-with-P} with $\prod_{k\geq 2}M_{k,1}$ we have the local complementary relation for the second party as
\begin{equation}
 \{M_{1,0},M_{1,1}\}\ket{\psi}=0.
 \label{eq:Alice-anti}
\end{equation}
For $k\geq 2$, multiplying the second equation of \cref{eq:raw-exact} by $M_{k,1}$ and using $M_{k,1}^2=\id$ we obtain
\begin{equation}
 M_{k,1}A_-\ket{\psi}=M_{1,1}\prod_{j\geq 2:j\ne k}M_{j,1}\ket{\psi}.
 \label{eq:T-transport}
\end{equation}
Then
\begin{align}
 M_{k,0}M_{k,1}A_-\ket{\psi}
 &=M_{1,1}\left(\prod_{j\geq 2:j\neq k}M_{j,1}\right)M_{k,0}\ket{\psi}\nonumber\\
 &\overset{\eqref{eq:S-chain-exact}}{=}M_{1,1}\left(\prod_{j\geq 2:j\ne k}M_{j,1}\right)M_{1,0}\ket{\psi}\nonumber\\
 &\overset{\eqref{eq:Alice-anti}}{=}-M_{1,0}M_{1,1}\left(\prod_{j\geq 2:j\ne k}M_{j,1}\right)\ket{\psi}\nonumber\\
 &\overset{\eqref{eq:T-transport}}{=}-M_{1,0}M_{k,1}A_-\ket{\psi}\nonumber\\
 &\overset{\eqref{eq:S-chain-exact}}{=}-A_-M_{k,1}M_{k,0}\ket{\psi}\nonumber\\
 \{M_{k,1},M_{k,0}\}A_-\ket{\psi} & = 0
 \label{eq:C-anti-transport}
\end{align}
where we use the fact that $A_{-}$ commutes with $M_{k,j}$ for $k\geq 1$. multiplying both sides with $A_{-}$ and using $A_{-}^2\ket{\psi}=\ket{\psi}$ we get the local complementary relations for the remaining parties
\begin{equation}
\{M_{k,0},M_{k,1}\}\ket{\psi}=0\quad (2\leq k\leq N-1).
 \label{eq:all-local-anti}
\end{equation}
Thus, we get the following result
\begin{lemma}
Let $\{\ket{\psi},A_{\pm}, (M_{i,0},M_{i,1}):i\geq 1\}$ constitutes the maximum quantum violation of the genuine network nonlocality functional given by \cref{eq:SN-compact}. The operators satisfy
\cref{eq:Bob-raw-anti,eq:Alice-anti,eq:all-local-anti} where $A_{0,y}=\frac{1}{\sqrt{2}}(A_+ +(-1)^{y}A_{-})$.
\end{lemma}

\subsection{Exact self-test: the local SWAP isometry}

Now, we are at a stage to prove the exact self-testing result for the nonlocality functional of \cref{eq:SN-compact}. To do that, we demonstrate an alternative method to represent the constraints \cref{eq:raw-exact,eq:S-chain-exact}. First note that, \cref{eq:S-chain-exact} can be written as follows: 
\begin{align}
\label{eq:S-stabilizers}
 M_{k,0}M_{k+1,0}\ket{\psi}&=\ket{\psi}\quad(1\le k\le N-2).
\end{align}
 We can have a similar relation for \cref{eq:raw-exact} by performing a procedure standard in self-testing \cite{BancalNavascues2015,Kaniewski2016,SupicBowlesReview}. 
 In particular, using the blockwise polar decomposition used in self-testing \cite{Kaniewski2016} we can get unitary counterparts of $A_{\pm}$. Since independent kernel completions need not preserve anticommutation\cite{SupicBowlesReview} we therefore use their common Jordan decomposition and choose the kernel extensions jointly. Trivial one-dimensional blocks are locally dilated to qubit blocks, as is standard in Jordan-lemma-based extraction constructions\cite{BancalNavascues2015,Kaniewski2016}. We refer to this joint choice as a coordinated polar decomposition. To put it precisely, Jordan's lemma tells us (see appendix \eqref{app:Jordan}) we can take 
\begin{align}
 A_+
 &=
 \bigoplus_\alpha
 \sqrt{2}\cos\theta_\alpha \tau_z^{(\alpha)},
 \label{eq:Hplus-Jordan-main}\\
 A_-
 &=
 \bigoplus_\alpha
 \sqrt{2}\sin\theta_\alpha\tau_x^{(\alpha)}.
 \label{eq:Hminus-Jordan}
\end{align}
where $\tau_x^{(\alpha)}$ and $\tau_z^{(\alpha)}$ are the Pauli matrices corresponding to the Jordan block $\mathcal{H}_{\alpha}$ and $0\leq \theta_{\alpha}\leq \frac{\pi}{2}$. To avoid zero-eigenvalue arises from the trivial Jordan block ($\theta_\alpha=0,\frac{\pi}{2}$) we can locally dilate the trivial Jordan block of $A_{\pm}$ by taking polar decomposition of each Block and obtain
\begin{equation}
\label{eq:Bob-completed}
\begin{split}
 M_{0,0}
 &:=
 \sgnbar(A_+)
 =
 \bigoplus_\alpha\tau_z^{(\alpha)},\\
 M_{0,1}
 &:=
 \sgnbar(A_-)
 =
 \bigoplus_\alpha\tau_x^{(\alpha)},
\end{split}
\end{equation}
Note that $M_{j,0}^2=\id$. Furthermore, by defining 
\begin{equation}
\label{eq:Apm-mod-operator}
\begin{split}
|A_{+}| &= \sqrt{A_{+}^{\dagger}A_{+}}= \bigoplus_{\alpha}\sqrt{2}\cos\theta_\alpha\id_{\alpha},\\
|A_{-}| &= \sqrt{A_{-}^{\dagger}A_{-}}= \bigoplus_{\alpha}\sqrt{2}\sin\theta_\alpha\id_{\alpha},
\end{split}
\end{equation}
we get $A_{\pm}^2=|A_{\pm}|^2$. We can also write 
\begin{equation}
\label{eq:Apm-polar-decomposition}
\begin{split}
A_{+} & =M_{0,0}|A_{+}|,\\
A_{-} & = M_{0,1}|A_{-}|.
\end{split}
\end{equation}

We define this procedure as coordinated polar decomposition. Before going further, we first note that following result 
\begin{lemma} 
\label{lem:inequality_coordinated_polar}
Let $X$ be Hermitian and let $\sgnbar(X)$ be a Hermitian reflection completing its polar sign on the kernel. If $Q$ is a reflection satisfying $[Q,X]=[Q,\sgnbar(X)]=0$, then
\begin{align}
 \norm{(Q-\sgnbar(X))\ket{\psi}}^2
 & \leq 4\norm{(Q-X)\ket{\psi}}^2.
 \label{eq:sign-operator}
\end{align}
\begin{align}
 \norm{(\sgnbar(X)-X)\ket{\psi}}^2
 & \leq \norm{(Q-X)\ket{\psi}}^2.
 \label{eq:nearest-sign-operator}
\end{align}
\end{lemma}
\begin{proof}
See appendix \ref{app:inequality_coordinated_polar}.
\end{proof}
In our application $X$ and $\sgnbar(X)$ are $A_{\pm}$ and $M_{0,j}$ with $j=0$ for $A_{+}$ and $j=1$ for $A_{-}$, whereas $Q$ acts on the other parties, so these commutation relations are automatic.

Note that,
$(M_{0,0},M_{0,1})$ satisfies
\begin{equation}
 M_{0,0}^2=M_{0,1}^2=\id,\quad \{M_{0,0},M_{0,1}\}=0.
 \label{eq:Bob-reflection-anti}
\end{equation}
For an operator $A$ satisfying \cref{eq:Hplus-onstate} or \cref{eq:Hminus-onstate} we can write 
\begin{equation}
\begin{split}
A^{2}\ket{\psi}=\ket{\psi}&,\\
(|A|^{2}-\id)\ket{\psi}&=0, \quad (A^{2}=|A|^2),\\
(|A|+\id)(|A|-&\id)\ket{\psi}=0,\\
(|A|-\id)\ket{\psi}&=0, \quad \text{(since $(|A|+\id)$ is invertible)}\\
|A|\ket{\psi}=\ket{\psi}&.
\end{split}
\end{equation}
This implies that \cref{eq:raw-exact} can be written as
\begin{equation}
 \label{eq:Bob-agreement}
\begin{split}
M_{1,0}\ket{\psi}&=A_+\ket{\psi}=M_{0,0}\ket{\psi},\\
 P\ket{\psi}&=A_-\ket{\psi}=M_{0,1}\ket{\psi}.
\end{split}
\end{equation}
Multiplying the second equality of \cref{eq:Bob-agreement} by $M_{0,1}$ gives
\begin{equation}
 G_T\ket{\psi}=\ket{\psi},
 \qquad
 G_T:=\prod_{i=0}^{N-1}M_{i,1}.
 \label{eq:global-T}
\end{equation}
In other words, $\ket{\psi}$ is a eigenstate of $G_T$ with eigenvalue $+1$ i.e. $G_{T}$ stabilizes $\ket{\psi}$. The first equality of \cref{eq:Bob-agreement} then gives the following
\begin{align}
\label{eq:S0S1stabilizers}
 M_{0,0}M_{1,0}\ket{\psi}&=\ket{\psi}.
\end{align}
Note that, \cref{eq:S-stabilizers,eq:S0S1stabilizers} are alternative forms of \cref{eq:raw-exact,eq:S-chain-exact}. However, \cref{eq:S-stabilizers,eq:S0S1stabilizers} are unitary on the full space and can be interpreted as hermitian reflection operators. We can write them in a compact form by constructing a graph $G=(V,E)$ with $|V|=N$ vertices and edge set, $E=\{(0,1),(1,2),(2,3),\ldots,(N-2,N-1)\}$. In the graph $G=(V,E)$ \cref{eq:S0S1stabilizers,eq:S-stabilizers} can be written as
\begin{equation}
\label{eq:tree-stabilizer}
M_{i,0}M_{j,0}\ket{\psi}=\ket{\psi}, \quad (i,j)\in E.
\end{equation}
 In other words, $\ket{\psi}$ is an eigenstate of $M_{i,0}M_{j,0}$ with eigenvalue $+1$. In other words, $M_{i,0}M_{j,0}$ is the stabilizer of $\ket{\psi}$. The local complementarity relations derived for the optimality conditions given by \cref{eq:Bob-raw-anti,eq:Alice-anti,eq:all-local-anti} are simply $\{M_{i,0},M_{i,1}\}\ket{\psi}=0$ for all $i$. The stabilizer relations in \cref{eq:global-T,eq:tree-stabilizer} are required for the optimal quantum violation. Now, we have all the machineries to construct the local isometry. 
 
For each site define
\begin{equation}
 \Pi_{a|i}:=\frac{\id+(-1)^aM_{i,0}}{2}\qquad(a=0,1).
 \label{eq:Pi-ai}
\end{equation}
Let the ancilla qubit for the $i$th quantum system is in $\ket{0}_i$.  The local isometry is the circuit
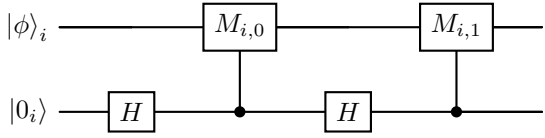
\begin{center}
\begin{quantikz}[row sep=0.55cm,column sep=0.65cm]
 \lstick{$\ket{\phi}_i$} & \qw & \gate{M_{i,0}} & \qw & \gate{M_{i,1}} & \qw\\
 \lstick{$\ket{0_i}$} & \gate{H} & \ctrl{-1} & \gate{H} & \ctrl{-1} & \qw
\end{quantikz}
\end{center}
defining the local unitary
\begin{equation}
 \label{eq:local-swap}
\begin{split}
\Phi_i(\ket{\phi}\ket{0}_i)
&=\Pi_{0|i}\ket{\phi}\ket{0}_i+M_{i,1}\Pi_{1|i}\ket{\phi}\ket{1}_i\\
&=\sum_{a=0}^{1}M_{i,1}^{a}\Pi_{a|i}\ket{\phi}\ket{a}_{i}.
\end{split}
\end{equation}
We can obtain \cref{eq:local-swap} from the sequence of gates described above as follows
\begin{align}
 \ket{\phi}\ket{0}_i
 &\xmapsto{H}
 \frac{\ket{\phi}\ket{0}_i+\ket{\phi}\ket{1}_i}{\sqrt{2}}\nonumber\\
 &\xmapsto{\mathrm c(M_{i,0})}
 \frac{\ket{\phi}\ket{0}_i+M_{i,0}\ket{\phi}\ket{1}_i}{\sqrt{2}}
 \nonumber\\
 &\xmapsto{H}
 \frac{(\mathbb{I}+M_{i,0})\ket{\phi}\ket{0}_i
       +(\mathbb{I}-M_{i,0})\ket{\phi}\ket{1}_i}{2}
 \nonumber\\
 &\xmapsto{\mathrm c(M_{i,1})}
 \Pi_{0|i}\ket{\phi}\ket{0}_i+M_{i,1}\Pi_{1|i}\ket{\phi}\ket{1}_i,
 \label{eq:swapCircuitDerivation}
\end{align}
where $\mathrm c(M_{i,y})$ is the controlled $M_{i,y}$ gate. Since, $M_{i,0}$ and $M_{i,1}$ are hermitian reflection operators, both controlled gates are unitary. Therefore, $\Phi_i$ is unitary. Consequently, \cref{eq:swapCircuitDerivation} proves that it restriction to inputs with ancilla $\ket{0}_i$ is an isometry. In particular, we can define the local isometry as
\begin{align}
\label{eq:isometry_i}
V_{i} &:=\sum_{a=0}^{1}M_{i,1}^{a}\Pi_{a|i}\otimes\ket{a}_{i},\\
\label{eq:isometry_dagger_i}
V_{i}^{\dagger} &:= \sum_{a=0}^{1}\Pi_{a|i}M_{i,1}^{a}\otimes\bra{a}_{i}
\end{align}
and write 
\begin{equation}
V_{i}\ket{\phi}=\sum_{a=0}^{1}M_{i,1}^{a}\Pi_{a|i}\ket{\phi}\otimes\ket{a}_{i}.
\end{equation}
We define the tensor product of local unitaries $\Phi=\bigotimes_{i=0}^{N-1}\Phi_i$ which gives,
\begin{equation}
 \Phi(\ket{\psi}\ket{0^N})
 =\sum_{\mathbf a}
 \left(\prod_{i=0}^{N-1}M_{i,1}^{a_i}\Pi_{a_i|i}\right)\ket{\psi}\ket{\mathbf a}.
 \label{eq:swap-expansion}
\end{equation}
where $\mathbf a=(a_0,\ldots,a_{N-1})\in\{0,1\}^N$.
If $a_i\ne a_j$ on any edge of the graph $G$, then using the cross-party commutation rule and $\Pi_{a_l|l}M_{l,0}=(-1)^{a_l}\Pi_{a_l|l}$ we have
\begin{align}
 \Pi_{a_i|i}\Pi_{a_j|j}\ket{\psi}
 &\overset{\eqref{eq:tree-stabilizer}}{=}\Pi_{a_i|i}\Pi_{a_j|j}M_{i,0}M_{j,0}\ket{\psi}\nonumber\\
 &=(-1)^{a_i+a_j}\Pi_{a_i|i}\Pi_{a_j|j}\ket{\psi},\nonumber\\
 &=-\Pi_{a_i|i}\Pi_{a_j|j}\ket{\psi},
 \label{eq:mixed-vanish}
\end{align}
so that amplitude vanishes.  Connectedness therefore leaves only the strings
$0^N$ and $1^N$:
\begin{align}
 \Phi(\ket{\psi}\ket{0^N})
 &=\left(\prod_i\Pi_{0|i}\right)\ket{\psi}\ket{0^N},\nonumber\\
  \label{eq:two-branches}
& +\left(\prod_iM_{i,1}\Pi_{1|i}\right)\ket{\psi}\ket{1^N}.
\end{align}
Now, we use
\begin{equation}
 (M_{i,1}\Pi_{1|i}-\Pi_{0|i}M_{i,1})\ket{\psi}
 =-\frac{1}{2}\{M_{i,1},M_{i,0}\}\ket{\psi}=0,
 \label{eq:move-T}
\end{equation}
and cross party commutations in the second term of \cref{eq:two-branches} to get
\begin{align}
 \left(\prod_iM_{i,1}\Pi_{1|i}\right)\ket{\psi}
 &=\left(\prod_i\Pi_{0|i}M_{i,1}\right)\ket{\psi}\nonumber\\
 &=\left(\prod_i\Pi_{0|i}\right)\left(\prod_{i}M_{i,1}\right)\ket{\psi}\nonumber\\
 &\overset{\eqref{eq:global-T}}{=}\left(\prod_i\Pi_{0|i}\right)G_T\ket{\psi},\nonumber\\
 &\overset{\eqref{eq:global-T}}{=}\left(\prod_i\Pi_{0|i}\right)\ket{\psi}.
 \label{eq:branches-equal}
\end{align}
By defining
\begin{equation}
 \ket{\xi}:=\sqrt{2}\left(\prod_i\Pi_{0|i}\right)\ket{\psi}.
 \label{eq:junk}
\end{equation}
and inserting \cref{eq:branches-equal} in \cref{eq:two-branches} we get
\begin{equation}
\Phi(\ket{\psi}\ket{0^N})
 =\ket{\xi}\otimes\ket{\GHZ_N}.
 \label{eq:exact-state}
\end{equation}
To demonstrate the self test of the measurement $M_{i,0}$ we consider the local unitary $\Phi_i$ defined in \cref{eq:local-swap} and obtain 
\begin{align}
\Phi_{i}(M_{i,0}\ket{\psi}\ket{0}_i) & = \Pi_{0|i}M_{i,0}\ket{\psi}\ket{0}_i\nonumber\\
& +M_{i,1}\Pi_{1|i}M_{i,0}\ket{\psi}\ket{1}_i\nonumber\\
& = \Pi_{0|i}\ket{\psi}\ket{0}_i-M_{i,1}\Pi_{1|i}\ket{\psi}\ket{1}_i\nonumber\\
& = \sum_{a=0}^{1}M_{i,1}^{a}\Pi_{a|i}\ket{\psi}\sigma_z^{(i)}\ket{a}_i\nonumber\\
 \label{eq:S-intertwine}
& = (\id\otimes\sigma_z^{(i)})\Phi_i(\ket{\psi}\ket{0_i}),
\end{align}
where we used  $\Pi_{0|i}M_{i,0}=\Pi_{0|i}$ and
$\Pi_{1|i}M_{i,0}=-\Pi_{1|i}$.
 To obtain the analogue of \cref{eq:S-intertwine} for $M_{i,1}$, we use the local complementary relations $\{M_{i,0},M_{i,1}\}\ket{\psi}=0$ and \cref{eq:Pi-ai} and obtain
\begin{equation}
 \label{eq:T-projector-identities}
\begin{split}
\Pi_{0|i}M_{i,1}\ket{\psi}&=M_{i,1}\Pi_{1|i}\ket{\psi},\\
 M_{i,1}\Pi_{1|i}M_{i,1}\ket{\psi}&=\Pi_{0|i}\ket{\psi}.
\end{split}
\end{equation}
Thus, we get
\begin{align}
\Phi_{i}(M_{i,1}\ket{\psi}\ket{0}_i) & = \Pi_{0|i}M_{i,1}\ket{\psi}\ket{0}_i\\
& \quad \quad +M_{i,1}\Pi_{1|i}M_{i,1}\ket{\psi}\ket{1}_i\nonumber\\
& = M_{i,1}\Pi_{1|i}\ket{\psi}\ket{0}_i+\Pi_{0|i}\ket{\psi}\ket{1}_i\nonumber\\
& = \sum_{a=0}^{1}M_{i,1}^{a}\Pi_{a|i}\ket{\psi}\sigma_x^{(i)}\ket{a}_i\nonumber\\
& = (\id\otimes\sigma_x^{(i)})\Phi_i(\ket{\psi}\ket{0}_i).
\label{eq:T-intertwine}
\end{align}
Equations \eqref{eq:S-intertwine}-\eqref{eq:T-intertwine} remain valid if $\ket{\psi}$ is first acted on by operators from other parties, because those operators commute with the local relations. Thus, we have the following result

\begin{theorem}
\label{thm:exact}
Let $\langle\mathcal S_N\rangle=\beta_N^{\rm Q}$ for the nonlocality functional of \cref{eq:SN-alt} for a strategy $\{\ket{\psi}, A_{0,y}, M_{i,x_i}:y\in\{0,1\}, 1\leq i\leq (N-1), x_{i}\in\{0,1\}\}$. For $j\in\{0,1\}$ let us define the physical observables 
\begin{equation}
O_{i,j}=\begin{cases}
    A_{0,j}, & \ i=0,\\
    M_{i,j}, & \ 1\leq i\leq N-1
\end{cases}
\end{equation}
The reference observables are given by 
\begin{equation}
O_{i,j}^{\mathrm{ref}}=\begin{cases}
    \frac{\sigma_z+(-1)^j\sigma_x}{\sqrt{2}} & \ i=0,\\
    \sigma^{i}_z & 1\leq i\leq N-1, \ j=0,\\
    \sigma^{i}_x & 1\leq i\leq N-1, \ j=1.
\end{cases}
\end{equation}
For any such strategy, there exists a local isometry whose unitary extension $\Phi$ satisfies \cref{eq:exact-state}.  Moreover, for every subset
$J\subseteq\{0,\ldots,N-1\}$ and every choice of input $x_i$ at each $i\in J$,
\begin{equation}
 \scalebox{0.75}{$\Phi\left[\left(\prod_{i\in J}O_{i,x_i}\right)
 \ket{\psi}\ket{0^N}\right]
 =\ket{\xi}\otimes
 \left(\prod_{i\in J}O_{i,x_i}^{\rm ref}\right)\ket{\GHZ_N}$}.
 \label{eq:exact-products}
\end{equation}
 In particular, for a single site
\begin{equation}
 \Phi(O_{i,j}\ket{\psi}\ket{0^N})
 =\ket{\xi}\otimes O_{i,j}^{\rm ref}\ket{\GHZ_N}.
 \label{eq:exact-single}
\end{equation}
\end{theorem}
\begin{proof}
The state statement is \cref{eq:exact-state}. At party $0$, use $A_{0,y}=(A_++(-1)^yA_-)/\sqrt2$ and the on-state agreement of $A_\pm$ with $M_{0,0}$ and $M_{0,1}$, respectively. For the other local measurements $(O_{i,0},O_{i,1})$,
\cref{eq:S-intertwine,eq:T-intertwine} and \cref{eq:exact-state} yields how $\Phi$ transforms the physical measurements of each party.  Because the same local calculation remains
valid after applying operators at other parties, successive application gives
\cref{eq:exact-products}.
\end{proof}

\section{Noise-robust analytic self-test}
\label{sec:robust}

Suppose that the physical realization satisfies
\begin{equation}
 \langle\mathcal S_N\rangle_\psi
 \geq \beta_N^{\rm Q}-\eps,
 \qquad \eps\geq0.
 \label{eq:robust-score}
\end{equation}
We use the same coordinated polar decompositions and local SWAP isometry as in the exact self-test. In particular,
$\{M_{0,0},M_{0,1}\}=0$ remains an exact operator identity, including away from the optimum. The other local anticommutation relations need only hold approximately on the physical state.

For the SOS residuals in \cref{eq:delta-def,eq:lambda-def}, define
\begin{equation}
 \begin{split}
 r_+&:=\norm{(M_{1,0}-A_+)\ket{\psi}},\\
 r_-&:=\norm{(P-A_-)\ket{\psi}},\\
 e_k&:=\norm{(M_{k,0}-M_{k+1,0})\ket{\psi}},
 \quad 1\leq k\leq N-2.
 \end{split}
 \label{eq:robust-residual-norms}
\end{equation}
The SOS identity gives the exact deficit decomposition
\begin{equation}
 \beta_N^{\rm Q}-\langle\mathcal S_N\rangle_\psi
 =\frac{r_+^2+r_-^2}{\sqrt2}
  +\sum_{k=1}^{N-2}e_k^2
 \leq\eps.
 \label{eq:exact-deficit-decomposition}
\end{equation}
We will bound the state fidelity directly in terms of these squared residual norms. Let the post-isometry vector and the reduced extracted state be
\begin{align}
 \ket{\Psi_{\rm sw}}&:=V\ket{\psi},
 \label{eq:Psi-sw}\\
 \rho_{\rm sw}&:=\operatorname{Tr}_{\xi}
 \bigl[\ketbra{\Psi_{\rm sw}}{\Psi_{\rm sw}}\bigr].
 \label{eq:rho-sw}
\end{align}
Here the trace discards the auxiliary systems. We use the squared fidelity with the pure target,
\begin{equation}
 F_N:=\bra{\GHZ_N}\rho_{\rm sw}\ket{\GHZ_N}.
 \label{eq:F-def}
\end{equation}
The target projector is
\begin{equation}
 \begin{split}
 \Gamma_N&:=\ketbra{\GHZ_N}{\GHZ_N}\\
 &=\frac{\id+\sigma_x^{\otimes N}}{2}
 \prod_{k=0}^{N-2}
 \frac{\id+\sigma_z^{(k)}\sigma_z^{(k+1)}}{2}.
 \end{split}
 \label{eq:GHZ-projector}
\end{equation}
On the physical systems, introduce
\begin{equation}
 \begin{split}
 C_a&:=\prod_{i=0}^{N-1}\Pi_{a|i},\qquad a\in\{0,1\},\\
 C&:=C_0+C_1
   =\prod_{k=0}^{N-2}
     \frac{\id+M_{k,0}M_{k+1,0}}{2},\\
 R_\pm&:=\frac{\id\pm G_T}{2},\qquad
 G_T=M_{0,1}P.
 \end{split}
 \label{eq:physical-consensus-projectors}
\end{equation}
 $C$ is a projector onto the joint $+1$ eigenspace of $M_{k,0}M_{k+1,0}$ for all $0\leq k\leq N-2$. The operators $R_\pm$ are also projectors on the orthonal subspaces of the Hermitian reflction $G_T$. Note that, \cref{eq:global-T,eq:tree-stabilizer} implies that the optimal $\ket{\psi}$ is the $+1$ eigenstate of $R_{+}$ and $C$ respectively. Away from maximal violation, these conditions need not hold. Consequently, we can quantify the failure of the global $G_T$ condition and of the joint $M_{i,0}M_{j,0}$ agreement condition given by\cref{eq:global-T,eq:tree-stabilizer} respectively as
\begin{equation}
 p:=\norm{R_-\ket{\psi}}^2,
 \qquad
 q:=\norm{(\id-C)\ket{\psi}}^2.
 \label{eq:pq-def}
\end{equation}
Now, we have the following result
\begin{lemma}[Exact pullback of the GHZ projector]
\label{lem:union_bound_prelim}
For the SWAP isometry of \cref{eq:isometry_i},
\begin{equation}
 \begin{split}
 K&:=V^\dagger(\id_\xi\otimes\Gamma_N)V=CR_+C,\\
 F_N&=\bra{\psi}K\ket{\psi}
     =\norm{R_+C\ket{\psi}}^2.
 \end{split}
 \label{eq:exact-GHZ-pullback}
\end{equation}
Consequently,
\begin{equation}
 \begin{split}
 1-F_N
 &\leq p+2q+2\sqrt{pq}\\
 &\leq\frac{3+\sqrt{5}}{2}(p+q).
 \end{split}
 \label{eq:union-bound-combined-prelim}
\end{equation}
No commutation between $C$ and $R_\pm$ is required.
\end{lemma}
\begin{proof}
See Appendix~\ref{appn:lem_union_bound_prelim}.
\end{proof}

Both failure probabilities are controlled by the SOS deficit. First, using $G_T=M_{0,1}P$ and $M_{0,1}^2=\id$ gives
\begin{equation}
 \norm{(G_T-\id)\ket{\psi}}
 =\norm{(P-M_{0,1})\ket{\psi}}.
 \label{eq:global-error}
\end{equation}
Therefore, the polar estimate \cref{eq:sign-operator} implies
\begin{equation}
 p=\frac{1}{4}\norm{(P-M_{0,1})\ket{\psi}}^2
 \leq r_-^2.
 \label{eq:p-bound}
\end{equation}
For any commuting projectors $Q_j$,
\begin{equation}
 \id-\prod_jQ_j\preceq\sum_j(\id-Q_j).
 \label{eq:projector-union}
\end{equation}
Applying this inequality to the commuting factors of $C$, we obtain
\begin{equation}
 \begin{split}
 q&\leq\frac{1}{4}\sum_{k=0}^{N-2}
   \norm{(M_{k,0}-M_{k+1,0})\ket{\psi}}^2\\
  &\leq r_+^2+\frac{1}{4}\sum_{k=1}^{N-2}e_k^2,
 \end{split}
 \label{eq:q-bound}
\end{equation}
where the $k=0$ term is bounded using \cref{eq:sign-operator} with $X=A_+$ and $Q=M_{1,0}$.
Combining \cref{eq:p-bound,eq:q-bound,eq:exact-deficit-decomposition},
\begin{equation}
 \begin{split}
 p+q&\leq r_+^2+r_-^2+\frac{1}{4}\sum_{k=1}^{N-2}e_k^2\\
 &\leq\sqrt{2}\left(
 \frac{r_+^2+r_-^2}{\sqrt{2}}+\sum_{k=1}^{N-2}e_k^2\right)
 \leq\sqrt{2}\eps.
 \end{split}
 \label{eq:pq-deficit-bound}
\end{equation}
Thus \cref{eq:union-bound-combined-prelim} yields
\begin{equation}
 1-F_N\leq\kappa\eps,
 \qquad \kappa:=\frac{3+\sqrt5}{\sqrt{2}}.
 \label{eq:state-fidelity}
\end{equation}

\begin{theorem}[Robust state self-test]
\label{thm:robust_state_selftest}
Let $N\geq3$ and suppose that $\langle\mathcal S_N\rangle_\psi\geq\beta_N^{\rm Q}-\eps$.
Define
\begin{equation}
 F_N^{\rm lb}:=\max\{0,1-\kappa\eps\},
 \qquad \kappa=\frac{3+\sqrt5}{\sqrt2}.
 \label{eq:lowebbound}
\end{equation}
Then the local SWAP isometry extracts an $N$-qubit state satisfying
\begin{equation}
 F_N\geq F_N^{\rm lb}.
 \label{eq:main-fidelity}
\end{equation}
\end{theorem}
\begin{proof}
Combine \cref{eq:state-fidelity} with $F_N\geq0$.
\end{proof}
The coefficient $\kappa$ is independent of $N$ when the error is expressed as the absolute deficit of the unnormalized functional $\mathcal S_N$. This does not imply an $N$-independent tolerance to every physical noise model. For example, consider mixing the reference state with white noise at visibility $v$ gives us the following physical state 
\begin{equation}
\label{eq:white_noise}
\rho_{v} = v\Gamma_N +(1-v)\frac{\id}{2^N}.
\end{equation}
For this state the nonlocality functional gives the following score for the reference measurements
\begin{equation}
\label{eq:SN_white_noise}
\beta^{Q}_{v}=v\beta^{Q}_{N}.
\end{equation}
Thus, for this noise model, the deficit is given by 
\begin{equation}
\label{eq:noise_tol_white_noise}
\epsilon(v)=\beta^{Q}_{N}-\beta^{Q}_{v}=(1-v)\beta^{Q}_{N}.
\end{equation}
Since $\beta^{Q}_{N}$ contains $N$-dependency $\epsilon(v)$ will also contain a $N$ dependency.

To obtain a common auxiliary state for the measurement-action bounds, define
\begin{equation}
 \delta_{st}:=\sqrt{2\left(1-\sqrt{F_N^{\rm lb}}\right)}.
 \label{eq:delta-state}
\end{equation}
Let
\begin{equation}
 \ket{\chi}:=(\id_\xi\otimes\bra{\GHZ_N})
 \ket{\Psi_{\rm sw}}.
 \label{eq:projected-junk}
\end{equation}
By \cref{eq:F-def,eq:rho-sw},
\begin{equation}
 \norm{\ket{\chi}}^2
 =\bra{\Psi_{\rm sw}}(\id_\xi\otimes\Gamma_N)
   \ket{\Psi_{\rm sw}}
 =F_N.
 \label{eq:junk-norm-fidelity}
\end{equation}
Whenever $F_N^{\rm lb}>0$, the normalized auxiliary state
\begin{equation}
 \ket{\xi_\eps}:=\frac{\ket{\chi}}{\sqrt{F_N}}
 \label{eq:normalized-robust-junk}
\end{equation}
is well defined, including at $F_N=1$. Its overlap with the extracted vector satisfies
\begin{equation}
 (\bra{\xi_\eps}\otimes\bra{\GHZ_N})\ket{\Psi_{\rm sw}}=\sqrt{F_N}.
\end{equation}
It follows that
\begin{equation}
 \begin{split}
 &\norm{\ket{\Psi_{\rm sw}}
       -\ket{\xi_\eps}\otimes\ket{\GHZ_N}}^2\\
 &\quad=2(1-\sqrt{F_N})\leq\delta_{st}^2,
 \end{split}
 \label{eq:vector-state-bound}
\end{equation}
In particular, the linear fidelity bound gives $\delta_{st}=O(\sqrt\eps)$ as $\eps\to0$.

We next control the local measurement actions. The exact identity \cref{eq:S-intertwine} holds for every input vector, so
\begin{equation}
 V(M_{i,0}\ket{\psi})
 =\sigma_z^{(i)}V\ket{\psi}
 =\sigma_z^{(i)}\ket{\Psi_{\rm sw}}.
 \label{eq:robust-S-exact}
\end{equation}
Consequently,
\begin{equation}
 \begin{split}
 &\norm{V(M_{i,0}\ket{\psi})
       -\ket{\xi_\eps}\otimes\sigma_z^{(i)}\ket{\GHZ_N}}\\
 &\qquad\leq\delta_{st}.
 \end{split}
 \label{eq:S-action-bound}
\end{equation}

For the other input, we use the following anticommutation estimate.
\begin{lemma}
\label{lem:error_from_incompatibility}
For $1\leq i\leq N-1$,
\begin{equation}
 \begin{split}
 \alpha_i&:=\norm{\{M_{i,0},M_{i,1}\}\ket{\psi}}\\
 &\leq2\norm{(P-M_{0,1})\ket{\psi}}
       +2\norm{(M_{i,0}-M_{0,0})\ket{\psi}}.
 \end{split}
 \label{eq:alpha-kl0}
\end{equation}
\end{lemma}
\begin{proof}
See Appendix~\ref{appn:error_from_incompatibility}.
\end{proof}
From \cref{eq:physical-consensus-projectors} we have the following result 
\begin{equation}
\label{eq:agree_on_C}
\begin{split}
M_{i,0}C_{a}&=(-1)^aC_{a}=M_{0,0}C_{a},\\
\implies&(M_{i,0}-M_{0,0})C=0.
\end{split}
\end{equation}
 In other words, $M_{i,0}$ and $M_{0,0}$ agree on the range of $C$ i.e. they will have same other outcomes. Consequently, for the projector that chooses different outcomes for the pair $M_{i,0}$ and $M_{0,0}$ given by
 \begin{equation}
\label{eq:diff_outcome_Mi0M01}
F_{i}=\frac{1}{2}\left(\id-M_{i,0}M_{0,0}\right)
 \end{equation}
 one has
 \begin{equation}
\begin{split}
F_{i}C&=\frac{1}{2}(C-M_{i,0}M_{0,0}C)=0,\\
CF_{i}&=\frac{1}{2}(C-M_{i,0}M_{0,0}C)=0.
\end{split}
 \end{equation}
 Therefore, $F_{i}$ and $C$ are orthogonal projectors, and their sum is also a projector and satisfies 
 \begin{equation}
C+F_i\preceq \id \implies D\preceq \id-C.
 \end{equation}
We then have
\begin{equation}
\begin{split}
 (M_{i,0}-M_{0,0})^2&=2\id-2M_{i,0}M_{0,0},\\
 &=4F_{i}\\
 &\preceq4(\id-C).
\end{split}
 \label{eq:consensus-local-disagreement}
\end{equation}
Thus
\begin{equation}
 \norm{(M_{i,0}-M_{0,0})\ket{\psi}}\leq2\sqrt q,
 \label{eq:uniform-di-bound}
\end{equation}
and \cref{eq:alpha-kl0,eq:p-bound} imply
\begin{equation}
 \alpha_i\leq4(\sqrt p+\sqrt q).
 \label{eq:uniform-alpha-bound}
\end{equation}
This estimate treats all sites $i\geq1$ uniformly.
\begin{lemma}
\label{lem:errorforMi1}
For $1\leq i\leq N-1$ and $F_N^{\rm lb}>0$,
\begin{equation}
 \begin{split}
 &\norm{V(M_{i,1}\ket{\psi})
       -\ket{\xi_\eps}\otimes\sigma_x^{(i)}\ket{\GHZ_N}}\\
 &\qquad\leq\delta_{st}+4\,2^{1/4}\sqrt\eps.
 \end{split}
 \label{eq:T-action-bound}
\end{equation}
\end{lemma}
\begin{proof}
See Appendix~\ref{app:errorforMi1}.
\end{proof}

For the distinguished party, the tested physical observables are
\begin{equation}
 A_{0,y}=\frac{A_++(-1)^yA_-}{\sqrt2},
 \qquad y\in\{0,1\},
 \label{eq:A0ydef}
\end{equation}
and their reference counterparts are
\begin{equation}
 A_{0,y}^{\rm ref}
 =\frac{\sigma_z^{(0)}+(-1)^y\sigma_x^{(0)}}{\sqrt2}.
 \label{eq:A0yref}
\end{equation}

\begin{lemma}
\label{lem:errorA0y}
For $y\in\{0,1\}$ and $F_N^{\rm lb}>0$,
\begin{equation}
 \begin{split}
 &\norm{V(A_{0,y}\ket{\psi})
       -\ket{\xi_\eps}\otimes A_{0,y}^{\rm ref}\ket{\GHZ_N}}\\
 &\qquad\leq\delta_{st}+2^{1/4}\sqrt\eps.
 \end{split}
 \label{eq:measurementerrorA0y}
\end{equation}
\end{lemma}
\begin{proof}
See Appendix~\ref{app:errorA0y}.
\end{proof}

These estimates use the same isometry and the same normalized auxiliary state, and extend to products with one tested observable per party.
\begin{theorem}[Robust state-and-measurement self-test]
\label{thm:robust-state-measurement-self-test}
Let $N\geq3$ and suppose that $\langle\mathcal S_N\rangle_\psi\geq\beta_N^{\rm Q}-\eps$.
Use $F_N^{\rm lb}$ and $\delta_{st}$ from \cref{eq:lowebbound,eq:delta-state}. Whenever $F_N^{\rm lb}>0$,
there exist local isometries $V_i$, with $V=\bigotimes_iV_i$, and a normalized state $\ket{\xi_\eps}$ such that
\begin{equation}
 \norm{V\ket{\psi}-\ket{\xi_\eps}\otimes\ket{\GHZ_N}}
 \leq\delta_{st}.
 \label{eq:robust-state-final}
\end{equation}
Define the local measurement costs
\begin{equation}
 \omega_{i,x}:=
 \begin{cases}
  2^{1/4}\sqrt\eps,&i=0,\quad x\in\{0,1\},\\[1mm]
  0,&1\leq i\leq N-1,\quad x=0,\\[1mm]
  4\,2^{1/4}\sqrt\eps,&1\leq i\leq N-1,\quad x=1.
 \end{cases}
 \label{eq:local-measurement-costs}
\end{equation}
For every subset $J\subseteq\{0,\ldots,N-1\}$ and every choice of inputs $\mathbf x_J=(x_i)_{i\in J}$,
\begin{equation}
 \begin{split}
 &\Bigg\|V\left[\left(\prod_{i\in J}O_{i,x_i}\right)
             \ket{\psi}\right]\\
 &\qquad-\ket{\xi_\eps}\otimes
       \left(\prod_{i\in J}O_{i,x_i}^{\rm ref}\right)
       \ket{\GHZ_N}\Bigg\|\\
 &\qquad\leq\delta_{st}+\sum_{i\in J}\omega_{i,x_i}.
 \end{split}
 \label{eq:robust-joint-actions-final}
\end{equation}
Here $O_{0,x}=A_{0,x}$, $O_{i,x}=M_{i,x}$ for $i\geq1$, and the reference observables are those of Theorem~\ref{thm:exact}.
\end{theorem}
\begin{proof}
The state statement is \cref{eq:vector-state-bound}. For the measurement statement, the exact $M_{i,0}$ intertwining identity
and the estimates proved in Appendices~\ref{app:errorforMi1} and~\ref{app:errorA0y} give
\begin{equation}
 \norm{V(O_{i,x}\ket{\psi})
       -O_{i,x}^{\rm ref}V\ket{\psi}}\leq\omega_{i,x}.
 \label{eq:local-intertwining-cost}
\end{equation}
To telescope a product, write its local error map as $D_i:=V_iO_{i,x_i}-(\id\otimes O_{i,x_i}^{\rm ref})V_i$. Operators at other input sites commute with $D_i^\dagger D_i$. Consequently, multiplying the input by tested reflections at those
sites leaves the norm of this local error unchanged. Tensoring with the other isometries and applying reference reflections also preserves the norm. Replacing the tested observables one site at a time therefore gives
\begin{equation}
 \begin{split}
 &\Bigg\|V\left[\left(\prod_{i\in J}O_{i,x_i}\right)
       \ket{\psi}\right]
       -\left(\prod_{i\in J}O_{i,x_i}^{\rm ref}\right)
        V\ket{\psi}\Bigg\|\\
 &\qquad\leq\sum_{i\in J}\omega_{i,x_i}.
 \end{split}
 \label{eq:joint-intertwining-cost}
\end{equation}
Finally, apply \cref{eq:vector-state-bound} once to the second vector. The product of reference reflections is unitary, so this last step contributes at most $\delta_{st}$.
\end{proof}
\section{Conclusion}

In summary, we have shown that the $N$-qubit GHZ state and the local measurement actions realizing a generating set of its stabilizer group can be self-tested using the genuine LOSR multipartite-nonlocality inequality introduced in Ref.~\cite{MaoEtAl2022}. Specifically, Theorem~\ref{thm:exact} establishes that every quantum realization attaining the maximal quantum value of the corresponding nonlocality functional is equivalent, up to local isometries and auxiliary degrees of freedom, to the reference $N$-qubit GHZ realization and its associated incompatible observables. We have further derived analytic robustness bounds for nonmaximal violations, including the linear fidelity bound $1-F_N\leq\kappa\eps$ with $\kappa=(3+\sqrt{5})/\sqrt{2}$, and $O(\sqrt{\eps})$ vector-norm bounds on the extracted state and the corresponding measurement actions. The Mao functional therefore provides a noise-robust, device-independent certification of the $N$-qubit GHZ state and the measurement actions generating its stabilizer correlations, without any assumption on the dimensions of the underlying Hilbert spaces.

At the level of the quantum realization being certified, our result is related to existing robust self-testing protocols. Stabilizer-based Bell inequalities have been used to self-test graph states and their local measurements \cite{BaccariEtAl2020}, while maximal violation of the Svetlichny inequality can self-test multipartite GHZ states and local anticommuting observables \cite{SinghSasmalPan2025}. The distinction of the present protocol therefore lies not in the target state alone, but in the operational meaning of the certifying functional. In addition to identifying the quantum realization, violation of the Mao inequality excludes causal explanations involving the simultaneous local composition of overlapping resources shared among at most $N-1$ parties, even when the parties are supplemented with shared randomness involving all $N$ parties.

Genuine LOSR multipartite nonlocality necessarily implies ordinary Bell nonlocality, since every Bell-local correlation can be generated using shared randomness alone and is therefore contained in the class of lower-order LOSR-producible correlations. It is, however, neither equivalent to nor generally ordered with genuine multipartite nonlocality in the historical Svetlichny sense \cite{CoiteuxRoyWolfeRenou2021PRL,CoiteuxRoyWolfeRenou2021PRA}. A Svetlichny-local model is a convex mixture over partitions and permits arbitrary, possibly signalling, correlations within the active groups. By contrast, the lower-order LOSR model permits the simultaneous local composition and joint processing of several overlapping causal resources but requires each constituent resource to arise within a causal generalized probabilistic theory. Because these two models enlarge the set of admissible explanations in different directions, neither corresponding notion of genuine multipartite nonlocality contains the other in general. The operational content of genuine LOSR multipartite nonlocality is that the observed correlations require an irreducibly $N$-partite nonclassical resource in any causal generalized probabilistic theory.

The present construction should also be distinguished from network-assisted self-testing \cite{SupicEtAl2023,SarkarOrtheyAugusiak2026}. In network-assisted protocols, a specified network topology and the independence of multiple sources form part of the assumptions used to extract the target state and measurements. No such independent-source decomposition of the physical quantum realization is assumed here. Instead, the network structure enters through the class of alternative causal explanations excluded by the Mao inequality. Consequently, an observed score above the genuine-LOSR bound and sufficiently close to the quantum maximum has two complementary implications: at the theory-agnostic level, it rules out simulation by arbitrary local compositions of lower-order causal resources; within quantum theory, it certifies that the physical realization is close, up to local isometries, to the $N$-qubit GHZ reference realization. This dual certification furnished by a single experimentally accessible nonlocality functional constitutes the principal conceptual contribution of our work.

\section*{Acknowledgements}

The author  acknowledges the help of ChatTGPT, DeepSeek, Claude as aids in exploring and verifying the segments of all mathematical proofs, polishing, condensing and editing this manuscript. The final writing, the mathematical results, their derivations, and their correctness were independently verified by the author, who takes full responsibility for the technical accuracy and integrity.

\bibliography{reference}

\appendix

\section{Jordan's lemma}
\label{app:Jordan}

Let us define the $+1$ spectral projectors of $A_{0,0}$ and $A_{0,1}$ defined in the Hilbert space $\mathcal{K}_{0}$ as
\begin{equation}
 P_0=\frac{\id+A_{0,0}}{2},
 \qquad
 P_1=\frac{\id+A_{0,1}}{2}.
 \label{eq:Jordan-projectors}
\end{equation}
Jordan's lemma gives an orthogonal decomposition of $\mathcal{K}_0$ into common invariant subspaces of $P_0$ and $P_1$,
\begin{equation}
 \mathcal{K}_0=\bigoplus_{\alpha}\mathcal{H}_\alpha,
 \qquad
 \dim\mathcal{H}_\alpha\le2.
 \label{eq:Jordan-direct-sum}
\end{equation}
Consequently,
\begin{equation}
 A_{0,y}=\bigoplus_\alpha A_{0,y}^{(\alpha)},
 \qquad y\in\{0,1\},
 \label{eq:B-direct-sum}
\end{equation}
where $A_{0,y}^{(\alpha)}$ denotes the restriction of $A_{0,y}$ to
$\mathcal H_\alpha$.

Consequently, we can write 
\begin{equation}
\begin{split}
P_0&=\bigoplus_{\alpha}\left(\frac{\id_{\alpha}+A_{0,0}^{(\alpha)}}{2}\right)=\bigoplus_{\alpha}P_0^{\alpha},\\
 P_1&=\bigoplus_{\alpha}\left(\frac{\id_{\alpha}+A_{0,1}^{(\alpha)}}{2}\right)=\bigoplus_{\alpha}P_1^{\alpha}.
\end{split}
 \label{eq:Jordan-projectors-block}
\end{equation}
For completeness, we will  derive the form of the observables on a
nontrivial two-dimensional block corresponding to $\mathcal{H}_{\alpha}$. Let us choose a normalized eigenvector of
$P_0P_1P_0$ on the range of $P_{0}$, $\operatorname{ran}P_0$ with non-zero support in $\mathcal{H}_{\alpha}$.  Thus, for a suitably eigenvector $\ket{u_\alpha}$ of $P_{0}$ satisfying $P_0\ket{u_\alpha}=\ket{u_\alpha}$ we have
\begin{equation}
P_0P_1P_0\ket{u_\alpha} =  \lambda_\alpha\ket{u_\alpha},
 \qquad
 0<\lambda_\alpha<1.
  \label{eq:Jordan-eigen}
\end{equation}
We can define a vector on the orthogonal subspace of $P_0$ in $\mathcal{H}_{\alpha}$ as
\begin{equation}
 \ket{v_\alpha}
 =
 \frac{
 P_1\ket{u_\alpha}
 -\lambda_\alpha\ket{u_\alpha}}
 {\sqrt{\lambda_\alpha(1-\lambda_\alpha)}}.
 \label{eq:Jordan-v}
\end{equation}
It follows that
\begin{equation}
 P_0\ket{v_\alpha}=0\implies
 \braket{u_\alpha|v_\alpha}=0,
\end{equation}
Consequently, in the basis
$\{\ket{u_\alpha},\ket{v_\alpha}\}$, we can define the parts of $P_{0}$ and $P_{1}$ in the Hilbert space $\mathcal{H}_\alpha$ as
\begin{equation}
 \label{eq:Jordan-matrices}
\begin{split}
 P_0^{(\alpha)}
 &=
 \begin{pmatrix}
  1&0\\
  0&0
 \end{pmatrix},\\
 P_1^{(\alpha)}
 &=
 \begin{pmatrix}
 \lambda_\alpha&
 \sqrt{\lambda_\alpha(1-\lambda_\alpha)}
 \\
 \sqrt{\lambda_\alpha(1-\lambda_\alpha)}&
 1-\lambda_\alpha
 \end{pmatrix}.
 \end{split}
\end{equation}
Writing
\begin{equation}
 \lambda_\alpha=\cos^2\theta_\alpha,
 \qquad
 0<\theta_\alpha<\frac{\pi}{2},
\end{equation}
and denoting the Pauli matrices on $\mathcal H_\alpha$ by $\sigma_z^{(\alpha)}$ and $\sigma_x^{(\alpha)}$, we obtain
\begin{align}
 A_{0,0}^{(\alpha)}
 &=2P_0^{(\alpha)}-\id_\alpha
 =\sigma_z^{(\alpha)},
 \label{eq:B0-Jordan-block}\\
 A_{0,1}^{(\alpha)}
 &=2P_1^{(\alpha)}-\id_\alpha \nonumber\\
 &=\cos(2\theta_\alpha)\sigma_z^{(\alpha)}
 +\sin(2\theta_\alpha)\sigma_x^{(\alpha)}.
 \label{eq:B1-Jordan-block}
\end{align}

It is convenient to introduce, separately on each Jordan block, the rotated Pauli pair
\begin{align}
 \tau_z^{(\alpha)}
 &=
 \cos\theta_\alpha\,\sigma_z^{(\alpha)}
 +\sin\theta_\alpha\,\sigma_x^{(\alpha)},
 \label{eq:tau-z-block}\\
 \tau_x^{(\alpha)}
 &=
 \sin\theta_\alpha\,\sigma_z^{(\alpha)}
 -\cos\theta_\alpha\,\sigma_x^{(\alpha)}.
 \label{eq:tau-x-block}
\end{align}
These obey
\begin{equation}
 \left(\tau_z^{(\alpha)}\right)^2
 =
 \left(\tau_x^{(\alpha)}\right)^2
 =\id_\alpha,
 \qquad
 \left\{
 \tau_z^{(\alpha)},\tau_x^{(\alpha)}
 \right\}=0.
 \label{eq:tau-Pauli-relations}
\end{equation}
Inverting Eqs.~\eqref{eq:tau-z-block}-\eqref{eq:tau-x-block} gives, on the block $\mathcal{H}_\alpha$,
\begin{align}
 A_{0,0}^{(\alpha)}
 &=
 \cos\theta_\alpha\,\tau_z^{(\alpha)}
 +\sin\theta_\alpha\,\tau_x^{(\alpha)},
 \label{eq:B0-tau-block}\\
 A_{0,1}^{(\alpha)}
 &=
 \cos\theta_\alpha\,\tau_z^{(\alpha)}
 -\sin\theta_\alpha\,\tau_x^{(\alpha)}.
 \label{eq:B1-tau-block}
\end{align}

Therefore the complete observables are
\begin{align}
 A_{0,0}
 &=
 \bigoplus_\alpha
 \left(
 \cos\theta_\alpha\,\tau_z^{(\alpha)}
 +\sin\theta_\alpha\,\tau_x^{(\alpha)}
 \right),
 \label{eq:B0-tau-direct-sum}\\
 A_{0,1}
 &=
 \bigoplus_\alpha
 \left(
 \cos\theta_\alpha\,\tau_z^{(\alpha)}
 -\sin\theta_\alpha\,\tau_x^{(\alpha)}
 \right).
 \label{eq:B1-tau-direct-sum}
\end{align}
Consequently, for
\begin{equation}
 A_+=\frac{A_{0,0}+A_{0,1}}{\sqrt{2}},
 \qquad
 A_-=\frac{A_{0,0}-A_{0,1}}{\sqrt{2}},
\end{equation}
we obtain
\begin{align}
 A_+
 &=
 \bigoplus_\alpha
 \sqrt{2}\cos\theta_\alpha\,
 \tau_z^{(\alpha)},
 \label{eq:Hplus-Jordan}\\
 A_-
 &=
 \bigoplus_\alpha
 \sqrt{2}\sin\theta_\alpha\,
 \tau_x^{(\alpha)}.
 \label{eq:Hminus-Jordan-appendix}
\end{align}

The one-dimensional Jordan blocks correspond to the endpoint cases $\theta_\alpha=0$ or $\theta_\alpha=\pi/2$.  Such a block may be harmlessly dilated to a two-dimensional local block, without changing the observed correlations, so that Eqs.~\eqref{eq:B0-tau-direct-sum}-\eqref{eq:Hminus-Jordan-appendix} continue to hold with $0\le\theta_\alpha\le\pi/2$.

For a Hermitian operator
\begin{equation}
 X=\sum_\lambda\lambda\Pi_\lambda,
\end{equation}
define
\begin{equation}
 |X|
 =
 \sum_\lambda|\lambda|\Pi_\lambda,
 \qquad
 U_X
 =
 \sum_{\lambda\neq0}
 \operatorname{sgn}(\lambda)\Pi_\lambda.
 \label{eq:polar}
\end{equation}
Thus
\begin{equation}
 X=U_X|X|.
\end{equation}
If $X$ has a nontrivial kernel, $U_X$ is a partial isometry that vanishes on $\ker X$. To obtain a Hermitian reflection,
we retain its action on $\operatorname{supp}|X|$ and choose a reflection on the kernel. For $A_+$ and $A_-$ these choices
are coordinated through their common Jordan decomposition.

To understand this in more detail, note that From Eqs.~\eqref{eq:Hplus-Jordan}-\eqref{eq:Hminus-Jordan-appendix}, on a block $\mathcal H_\alpha$,
\begin{align}
 A_+^{(\alpha)}
 &=
 \sqrt{2}\cos\theta_\alpha\,
 \tau_z^{(\alpha)},
 \label{eq:Hplus-block-polar}\\
 A_-^{(\alpha)}
 &=
 \sqrt{2}\sin\theta_\alpha\,
 \tau_x^{(\alpha)}.
 \label{eq:Hminus-block-polar}
\end{align}
Hence, whenever the corresponding coefficient is nonzero,
\begin{align}
 \operatorname{sgn}
 \left(A_+^{(\alpha)}\right)
 &=
 \tau_z^{(\alpha)}
 \qquad
 (\cos\theta_\alpha>0),
 \label{eq:sign-Hplus-block}\\
 \operatorname{sgn}
 \left(A_-^{(\alpha)}\right)
 &=
 \tau_x^{(\alpha)}
 \qquad
 (\sin\theta_\alpha>0).
 \label{eq:sign-Hminus-block}
\end{align}

There are two endpoint cases.  If $\theta_\alpha=0$, then
\begin{equation}
 A_+^{(\alpha)}
 =
 \sqrt{2}\,\tau_z^{(\alpha)},
 \qquad
 A_-^{(\alpha)}=0.
\end{equation}
The polar part of $A_-^{(\alpha)}$ is therefore undefined on this block, and we complete it by choosing
\begin{equation}
 \overline{\operatorname{sgn}}
 \left(A_-^{(\alpha)}\right)
 :=\tau_x^{(\alpha)}.
\end{equation}
Similarly, if $\theta_\alpha=\pi/2$, then
\begin{equation}
 A_+^{(\alpha)}=0,
 \qquad
 A_-^{(\alpha)}
 =
 \sqrt{2}\,\tau_x^{(\alpha)},
\end{equation}
and we choose
\begin{equation}
 \overline{\operatorname{sgn}}
 \left(A_+^{(\alpha)}\right)
 :=\tau_z^{(\alpha)}.
\end{equation}

Thus the completed signs may be defined uniformly, block by block, as
\begin{align}
 \overline{\operatorname{sgn}}
 \left(A_+^{(\alpha)}\right)
 &:=
 \tau_z^{(\alpha)},
 \label{eq:completed-Hplus-block}\\
 \overline{\operatorname{sgn}}
 \left(A_-^{(\alpha)}\right)
 &:=
 \tau_x^{(\alpha)}.
 \label{eq:completed-Hminus-block}
\end{align}
Taking the direct sum over all Jordan blocks gives the two completed Bob reflections
\begin{align}
 M_{0,0}
 &:=
 \sgnbar(A_+)
 =
 \bigoplus_\alpha\tau_z^{(\alpha)},
 \label{eq:S1-completed-app}\\
 M_{0,1}
 &:=
 \sgnbar(A_-)
 =
 \bigoplus_\alpha\tau_x^{(\alpha)}.
 \label{eq:T1-completed-app}
\end{align}
Their reflection and anticommutation properties now follow blockwise:
\begin{align}
 M_{0,0}^2
 &=
 \bigoplus_\alpha
 \left(\tau_z^{(\alpha)}\right)^2
 =\id,
 \label{eq:S1-reflection}\\
 M_{0,1}^2
 &=
 \bigoplus_\alpha
 \left(\tau_x^{(\alpha)}\right)^2
 =\id,
 \label{eq:T1-reflection}\\
 \{M_{0,0},M_{0,1}\}
 &=
 \bigoplus_\alpha
 \left\{
 \tau_z^{(\alpha)},\tau_x^{(\alpha)}
 \right\}
 =0.
 \label{eq:S1T1-anticommutation}
\end{align}

Note that the coordinated polar decomposition works for arbitrary finite dimensions.

\section{Proof of Lemma~\ref{lem:inequality_coordinated_polar}}
\label{app:inequality_coordinated_polar}

Let $X$ be Hermitian and let $\sgnbar(X)$ be a Hermitian reflection that agrees with $\operatorname{sgn}(X)$ on the support of $X$. Let $Q$ be a reflection commuting with both operators. Since $X$, $Q$, and $\sgnbar(X)$ commute, use their joint spectral decomposition to write
\begin{equation}
 \begin{split}
 X\Pi_\alpha&=t_\alpha\Pi_\alpha,\\
 Q\Pi_\alpha&=q_\alpha\Pi_\alpha,\\
 \sgnbar(X)\Pi_\alpha&=r_\alpha\Pi_\alpha,
 \end{split}
 \label{eq:common-spectral-decomp}
\end{equation}
where $t_\alpha\in\mathbb R$ and $q_\alpha,r_\alpha\in\{-1,1\}$. For $t_\alpha\neq0$, $r_\alpha=\operatorname{sgn}(t_\alpha)$;
for $t_\alpha=0$, either sign can occur. In all cases, $t_\alpha=r_\alpha|t_\alpha|$. 

If $q_\alpha=r_\alpha$, then $|q_\alpha-r_\alpha|=0$ and $|r_\alpha-t_\alpha|=|q_\alpha-t_\alpha|$. If
$q_\alpha=-r_\alpha$, then
\begin{equation}
 \begin{split}
 |q_\alpha-r_\alpha|&=2,\\
 |q_\alpha-t_\alpha|&=1+|t_\alpha|,\\
 |r_\alpha-t_\alpha|&=|1-|t_\alpha||.
 \end{split}
\end{equation}
Hence, in both cases,
\begin{align}
 |q_\alpha-r_\alpha|&\leq2|q_\alpha-t_\alpha|,
 \label{eq:sign-scalar}\\
 |r_\alpha-t_\alpha|&\leq|q_\alpha-t_\alpha|.
 \label{eq:nearest-sign-scalar}
\end{align}
For $\ket{\psi_\alpha}:=\Pi_\alpha\ket{\psi}$, orthogonality of
the spectral sectors gives
\begin{equation}
 \begin{split}
 \norm{(Q-\sgnbar(X))\ket{\psi}}^2
 &=\sum_\alpha|q_\alpha-r_\alpha|^2
   \norm{\ket{\psi_\alpha}}^2\\
 &\leq4\norm{(Q-X)\ket{\psi}}^2,
 \end{split}
\end{equation}
and
\begin{equation}
 \begin{split}
 \norm{(\sgnbar(X)-X)\ket{\psi}}^2
 &=\sum_\alpha|r_\alpha-t_\alpha|^2
   \norm{\ket{\psi_\alpha}}^2\\
 &\leq\norm{(Q-X)\ket{\psi}}^2.
 \end{split}
\end{equation}
These are \cref{eq:sign-operator,eq:nearest-sign-operator}. 

\section{Proof of Lemma~\ref{lem:union_bound_prelim}}
\label{appn:lem_union_bound_prelim}

Note that the local SWAP isometry of \cref{eq:isometry_i} satisfies
\begin{equation}
 (\mathbb I_{\mathcal K_i^\xi}\otimes\langle a|_i)V_i =M_{i,1}^{a}\Pi_{a|i},\quad a\in\{0,1\}.
\end{equation}
Tensoring these identities over all parties, with the identity on the purification system understood, gives
\begin{equation}
 (\mathbb I_\xi\otimes\langle\mathbf a|)V
 =\prod_{i=0}^{N-1}M_{i,1}^{a_i}\Pi_{a_i|i}.
\end{equation}
Taking the GHZ superposition of the all-zero and all-one strings therefore yields
\begin{equation}
\begin{split}
    W:=(\id_\xi&\otimes\bra{\GHZ_N})V\\
    &=\left(
   \prod_{i=0}^{N-1}\Pi_{0|i}
   +\prod_{i=0}^{N-1}M_{i,1}\Pi_{1|i}
 \right)\\
 &  =\frac{C_0+G_TC_1}{\sqrt{2}}.
\end{split}
 \label{eq:GHZ-component-map}
\end{equation}

 Consequently, using $C_a^2=C_a$ and $G_T^2=\id$,
\begin{equation}
 \begin{split}
  W^{\dagger}W
 &=\frac{1}{2}\bigl(C_0+C_1+C_0G_TC_1+C_1G_TC_0\bigr).
 \end{split}
 \label{eq:GHZ-pullback-expanded}
\end{equation}
Exact anticommutation at party $0$ $\{M_{0,1},M_{0,0}\}=0$ implies $\Pi_{a|0}M_{0,1}\Pi_{a|0}=M_{0,1}\Pi_{a+1|0}\Pi_{a|0}=0$ where $a+1$ is addition modulo two. All operators at the other parties commute with $\Pi_{a|0}$ and $M_{0,1}$. Therefore
\begin{equation}
 C_aG_TC_a=0,\qquad a\in\{0,1\}.
 \label{eq:zero-diagonal-GT-blocks}
\end{equation}
Expanding $CR_+C=C(\id+G_T)C/2$ and using \cref{eq:zero-diagonal-GT-blocks} proves $W^{\dagger}W=CR_+C$. Since $R_+$ is a projector, this yields \cref{eq:exact-GHZ-pullback}.

To estimate the fidelity, use 
\begin{equation}
\begin{split}
C\ket{\psi}&=R_{-}C\ket{\psi}+R_{+}C\ket{\psi},\\
\norm{C\ket{\psi}}^2 &=\norm{R_{-}C\ket{\psi}}^2+\norm{R_{+}C\ket{\psi}}^2
\end{split}
\end{equation}
and obtain
\begin{equation}
 \begin{split}
 1-F_N
 &=1-\norm{C\ket{\psi}}^2
   +\norm{R_-C\ket{\psi}}^2\\
 &=q+\norm{R_-C\ket{\psi}}^2.
 \end{split}
 \label{eq:fidelity-exact-decomposition}
\end{equation}
The identity $R_-C\ket{\psi}=R_-\ket{\psi}-R_-(\id-C)\ket{\psi}$, the triangle inequality, and $\norm{R_-}_\infty\leq1$ imply
\begin{equation}
 \norm{R_-C\ket{\psi}}\leq\sqrt p+\sqrt q.
 \label{eq:compressed-phase-error}
\end{equation}
Thus
\begin{equation}
 \begin{split}
 1-F_N&\leq q+(\sqrt p+\sqrt q)^2\\
 &=p+2q+2\sqrt{pq}.
 \end{split}
\end{equation}
Finally,
\begin{equation}
 \begin{split}
 p+2q+2\sqrt{pq}
 &=\begin{pmatrix}\sqrt p&\sqrt q\end{pmatrix}
   \begin{pmatrix}1&1\\1&2\end{pmatrix}
   \begin{pmatrix}\sqrt p\\\sqrt q\end{pmatrix}\\
 &\leq\frac{3+\sqrt{5}}{2}(p+q),
 \end{split}
\end{equation}
because the largest eigenvalue of the displayed matrix is $(3+\sqrt5)/2$. This proves \cref{eq:union-bound-combined-prelim} without imposing any commutation relation between $C$ and $R_\pm$.

\section{Proof of Lemma~\ref{lem:error_from_incompatibility}}
\label{appn:error_from_incompatibility}

For $i=1,\ldots,N-1$, define the reflection
\begin{equation}
 R_i:=M_{i,1}G_T=\prod_{j\neq i}M_{j,1}.
 \label{eq:Rk-def}
\end{equation}
It acts only at sites other than $i$. Therefore,
\begin{equation}
 [R_i,M_{i,0}]=0,
 \qquad [M_{i,1},M_{0,0}]=0.
\end{equation}
Moreover,
\begin{equation}
 (R_i-M_{i,1})\ket{\psi}
 =M_{i,1}(G_T-\id)\ket{\psi},
 \label{eq:Rk-error-vector}
\end{equation}
so \cref{eq:global-error} gives
\begin{equation}
 \norm{(R_i-M_{i,1})\ket{\psi}}
 =\norm{(P-M_{0,1})\ket{\psi}}.
 \label{eq:Rk-error}
\end{equation}
Because the factor at site $0$ in $R_i$ is $M_{0,1}$, its exact anticommutation with $M_{0,0}$ gives
\begin{equation}
 \{R_i,M_{0,0}\}
 =\left(\prod_{j\neq0,i}M_{j,1}\right)
   \{M_{0,1},M_{0,0}\}=0.
 \label{eq:RkMl0}
\end{equation}
Using these commutation and anticommutation identities, expand
\begin{equation}
 \begin{split}
 \{M_{i,0},M_{i,1}\}
 ={}&M_{i,0}(M_{i,1}-R_i)\\
 &+R_i(M_{i,0}-M_{0,0})\\
 &+M_{0,0}(M_{i,1}-R_i)\\
 &+M_{i,1}(M_{i,0}-M_{0,0}).
 \end{split}
 \label{eq:alpha-k-expansion}
\end{equation}
Applying this identity to $\ket{\psi}$, the triangle inequality and the unitarity of the left factors imply
\begin{equation}
 \begin{split}
 \norm{\{M_{i,0},M_{i,1}\}\ket{\psi}}
 &\leq2\norm{(M_{i,1}-R_i)\ket{\psi}}\\
 &\quad+2\norm{(M_{i,0}-M_{0,0})\ket{\psi}}.
 \end{split}
 \label{eq:alpha-k}
\end{equation}
Substituting \cref{eq:Rk-error} proves \cref{eq:alpha-kl0}.

\section{Proof of Lemma~\ref{lem:errorforMi1}}
\label{app:errorforMi1}

For a local input vector, with identities on all other systems understood, the SWAP isometry gives
\begin{equation}
 \begin{split}
 &\left[V_iM_{i,1}-(\id\otimes\sigma_x^{(i)})V_i\right]
 \ket{\psi}\\
 &\quad=\frac{1}{2}\{M_{i,0},M_{i,1}\}\ket{\psi}\ket0\\
 &\qquad-\frac{1}{2}M_{i,1}\{M_{i,0},M_{i,1}\}
                \ket{\psi}\ket1.
 \end{split}
\end{equation}
The two terms occupy orthogonal ancilla sectors and $M_{i,1}$ is unitary. Tensoring with the other local isometries preserves
the norm, so
\begin{equation}
 \begin{split}
 &\norm{V(M_{i,1}\ket{\psi})-\sigma_x^{(i)}V\ket{\psi}}\\
 &\qquad=\frac{\alpha_i}{\sqrt2}.
 \end{split}
 \label{eq:sigmaxerror}
\end{equation}
Now \cref{eq:uniform-alpha-bound,eq:pq-deficit-bound} imply
\begin{equation}
 \begin{split}
 \frac{\alpha_i}{\sqrt2}
 &\leq2\sqrt2(\sqrt p+\sqrt q)\\
 &\leq4\sqrt{p+q}\leq4\,2^{1/4}\sqrt\eps.
 \end{split}
 \label{eq:robust-T-intertwine}
\end{equation}
Adding and subtracting $\sigma_x^{(i)}\ket{\Psi_{\rm sw}}$ in the desired measurement-action difference gives
\begin{equation}
 \begin{split}
 &\norm{V(M_{i,1}\ket{\psi})
       -\ket{\xi_\eps}\otimes\sigma_x^{(i)}\ket{\GHZ_N}}\\
 &\quad\leq
   \norm{V(M_{i,1}\ket{\psi})-\sigma_x^{(i)}V\ket{\psi}}\\
 &\qquad+
   \norm{\ket{\Psi_{\rm sw}}-\ket{\xi_\eps}\otimes\ket{\GHZ_N}}\\
 &\quad\leq4\,2^{1/4}\sqrt\eps+\delta_{st},
 \end{split}
\end{equation}
where the Pauli operator preserves the vector norm. This proves
\cref{eq:T-action-bound}.

\section{Proof of Lemma~\ref{lem:errorA0y}}
\label{app:errorA0y}

Define the completed reference combination
\begin{equation}
 \widetilde A_{0,y}:=
 \frac{M_{0,0}+(-1)^yM_{0,1}}{\sqrt2}.
 \label{eq:tildeA0y}
\end{equation}
The exact operator relation $\{M_{0,0},M_{0,1}\}=0$ gives exact intertwining for both completed reflections at party $0$, for
every input vector. Thus
\begin{equation}
 \begin{split}
 VM_{0,0}&=\sigma_z^{(0)}V,\\
 VM_{0,1}&=\sigma_x^{(0)}V,\\
 V\widetilde A_{0,y}&=A_{0,y}^{\rm ref}V.
 \end{split}
 \label{eq:isometryAoytilde}
\end{equation}
By \cref{eq:nearest-sign-operator},
\begin{equation}
 \begin{split}
 \norm{(A_+-M_{0,0})\ket{\psi}}&\leq r_+,\\
 \norm{(A_--M_{0,1})\ket{\psi}}&\leq r_-.
 \end{split}
 \label{eq:completed-raw-residuals}
\end{equation}
Consequently, the triangle and Cauchy-Schwarz inequalities give
\begin{equation}
 \begin{split}
 \norm{(A_{0,y}-\widetilde A_{0,y})\ket{\psi}}
 &\leq\frac{r_++r_-}{\sqrt2}\\
 &\leq\sqrt{r_+^2+r_-^2}
 \leq2^{1/4}\sqrt\eps,
 \end{split}
 \label{eq:Bob-raw-robust}
\end{equation}
where the last step uses \cref{eq:exact-deficit-decomposition}. Since $V$ is an isometry, \cref{eq:isometryAoytilde} also yields
\begin{equation}
 \begin{split}
 &\norm{V(A_{0,y}\ket{\psi})-A_{0,y}^{\rm ref}V\ket{\psi}}\\
 &\qquad=\norm{(A_{0,y}-\widetilde A_{0,y})\ket{\psi}}
 \leq2^{1/4}\sqrt\eps.
 \end{split}
 \label{eq:raw-first-party-intertwining}
\end{equation}
Finally, $A_{0,y}^{\rm ref}$ is a reflection, so
\begin{equation}
 \begin{split}
 &\norm{V(A_{0,y}\ket{\psi})
       -\ket{\xi_\eps}\otimes A_{0,y}^{\rm ref}\ket{\GHZ_N}}\\
 &\quad\leq
 \norm{V(A_{0,y}\ket{\psi})-A_{0,y}^{\rm ref}V\ket{\psi}}\\
 &\qquad+
 \norm{\ket{\Psi_{\rm sw}}-\ket{\xi_\eps}\otimes\ket{\GHZ_N}}\\
 &\quad\leq2^{1/4}\sqrt\eps+\delta_{st}.
 \end{split}
\end{equation}
This proves \cref{eq:measurementerrorA0y} with the same auxiliary state used in the state and all other measurement-action bounds.

\end{document}